\documentclass{llncs}

\usepackage{figlatex}\graphicspath{{fig/}}

\usepackage{subfigure}
\usepackage[english]{babel}
\usepackage{color,amssymb, url,amsfonts, amsmath}	
\usepackage{tikz,xcolor,pgf,verbatim,xspace}%fourier,,enumitem
\usepackage{algpseudocode}
\usepackage{todonotes}
\usepackage{paralist}
\usepackage[utf8]{inputenc}
\newcommand{\qrm}{\texttt{qr\_mumps}\xspace}
\newcommand{\mumps}{\texttt{MUMPS}\xspace}
\newcommand{\superlu}{\texttt{SuperLU}\xspace}
\newcommand{\pastix}{\texttt{PaStiX}\xspace}

\usetikzlibrary{calc, fit, shapes,arrows}
\usetikzlibrary{decorations,decorations.pathreplacing,patterns,shapes.multipart}
\pgfdeclarelayer{background}
\pgfdeclarelayer{foreground}
\pgfsetlayers{background,main,foreground}

\newcommand{\inte}[4][]{\int_{#2}^{#1} \! #3 \ \mathrm #4}
\def\mystrut(#1){\vrule height #1 depth 0pt width 0pt}

	\newcommand{\newparskip}{\bigskip}

        \newtheorem{defi}[theorem]{Definition}

\newcommand{\LG}[1]{\mathcal{L}_{#1}}
\newcommand{\para}[2]{#1 \mathop{\parallel} #2}
\newcommand{\seri}[2]{#1 \mathop{;} #2}
\newcommand{\s}{\mathcal S\xspace}
\newcommand{\spm}{\ensuremath{\mathcal{S}_{\mathrm{PM}}}\xspace}
\newcommand{\mpm}{\ensuremath{M^{\mathrm{PM}}}\xspace}

\newcommand{\sopt}{\ensuremath{\mathcal{S}_{\mathrm{OPT}}}\xspace}
\newcommand{\uopt}{\ensuremath{u_{\mathrm{OPT}}}\xspace}

\newcommand{\mopt}{\ensuremath{{M}_{\mathrm{OPT}}}\xspace}

\newcommand{\A}{\ensuremath{\mathcal A}\xspace}
\newcommand{\B}{\ensuremath{\mathcal B}\xspace}

\newcommand{\T}{\Delta}

\newcommand{\tl}{\tau}
\newcommand{\R}{\ensuremath{\mathcal{R}}\xspace}
\newcommand{\rat}{r}
\newcommand{\fctofproc}{{processor profile}\xspace}
\newcommand{\frtrd}{\ensuremath{\left(\frac{4}{3}\right)^\alpha}}
\newcommand{\mideal}{\ensuremath{{M}_{\mathrm{ideal}}}\xspace}
\newcommand{\saideal}{\ensuremath{{\Sigma}_{\mathrm{ideal}}}\xspace}
\newcommand{\sa}{\ensuremath{\sum_{i\in A} x_i}}
\newcommand{\ao}{\ensuremath{A_{\mathrm{O}}}}
\newcommand{\sab}{\ensuremath{\sum_{i\in \bar A} x_i}}
\newcommand{\sao}{\ensuremath{\sum_{i\in \ao} x_i}}
\newcommand{\saob}{\ensuremath{\sum_{i\in \bar \ao} x_i}}
\newcommand{\saalp}{\ensuremath{\left(\sum_{i\in A} x_i\right)^\alpha}}
\newcommand{\sabalp}{\ensuremath{\left(\sum_{i\in \bar A} X_i\right)^\alpha}}
\newcommand{\saopt}{\ensuremath{\mathcal{\Sigma}_{\mathrm{OPT}}}\xspace}
\newcommand{\puisa}[1]{\left(#1\right)^\alpha}
\newcommand{\eg}{\varepsilon_\gamma}
\newcommand{\ek}{\varepsilon_{\kappa}}
\newcommand{\el}{\varepsilon_\lambda}
\DeclareMathOperator*{\argmin}{arg\,min}
\newcommand{\midset}{\mathrel{}\middle|\mathrel{}}

\newcommand{\pqres}{ {\sc $(p,q)$-Scheduling Restricted}\xspace}
\newcommand{\pqsched}{ {\sc $(p,q)$-Scheduling}\xspace}

\newif\iflong
\longfalse

\makeatletter
\def\blfootnote{\gdef\@thefnmark{}\@footnotetext}
\makeatother

\begin{document}
\title{Scheduling Trees of Malleable Tasks\\ for Sparse Linear
  Algebra}

\titlerunning{Scheduling Trees of Malleable Tasks}

\author{Abdou Guermouche\inst{1} \and Loris Marchal\inst{2} \and Bertrand Simon\inst{2} \and Fr\'ed\'eric Vivien\inst{2}}
\institute{University of Bordeaux and INRIA,\\
     200 rue de la vieille Tour,  Talence, France\\
     \url{abdou.guermouche@labri.fr}
\and
CNRS, INRIA and University of Lyon,\\
     LIP, ENS Lyon, 46 allée d'Italie, Lyon, France\\
     \url{{loris.marchal,bertrand.simon,frederic.vivien}@ens-lyon.fr}}

\maketitle

\begin{abstract}
  Scientific\blfootnote{This work was supported by the ANR SOLHAR project
  funded by the French National Research Agency.} workloads are often described by directed acyclic task
  graphs. This is in particular the case for multifrontal
  factorization of sparse matrices ---the focus of this paper--- whose
  task graph is structured as a tree of parallel tasks. Prasanna and
  Musicus~\cite{prasmus,prasmus2} advocated using the concept of
  \emph{malleable} tasks to model parallel tasks involved in matrix
  computations.  In this powerful model each task is processed on a
  time-varying number of processors.  Following Prasanna and Musicus,
  we consider malleable tasks whose speedup is $p^\alpha$, where $p$
  is the fractional share of processors on which a task executes, and
  $\alpha$ ($0 < \alpha \leq 1$) is a task-independent parameter.
  Firstly, we use actual experiments on multicore platforms to
  motivate the relevance of this model for our application. Then, we
  study the optimal time-minimizing allocation proposed by Prasanna
  and Musicus using optimal control theory. We greatly simplify their
  proofs by resorting only to pure scheduling arguments. Building on
  the insight gained thanks to these new proofs, we extend the study
  to distributed (homogeneous or heterogeneous) multicore
  platforms. We prove the NP-completeness of the corresponding
  scheduling problem, and we then propose some approximation
  algorithms.
\end{abstract}

% \begin{IEEEkeywords}
%   Scheduling; task graphs; malleable tasks; sparse linear algebra.
% \end{IEEEkeywords}

\section{Introduction}
\label{sec:intro}

Parallel workloads are often modeled as directed acyclic task graphs,
or DAGs, where nodes represent tasks and edges represent dependencies
between tasks. Task graphs arise from many scientific domains, such as
image processing, genomics, and geophysical simulations. In this
paper, we focus on task graphs coming from sparse linear algebra, and
especially from the factorization of sparse matrices using the
multifrontal method. Liu~\cite{liu:90} explains that the computational
dependencies and requirements in Cholesky and LU factorization of
sparse matrices using the multifrontal method can be modeled as a task
tree, called the \emph{assembly tree}. We therefore focus on
dependencies that can be modeled as a tree.

In the abundant existing literature, several variants of the task
graph scheduling problem are addressed, depending on the ability to
process a task in parallel: tasks are either \emph{sequential} (not
amenable to parallel processing), \emph{rigid} (requesting a given
number of processors), \emph{moldable} (able to cope with any fixed
number of processors) or even \emph{malleable} (processed on a variable
number of processors) in the terminology of~Drozdowski~\cite[chapter
25]{handbook}. When considering moldable and malleable tasks, one has
to define how the processing time of a task depends on the number of
allocated  processors.  Under some general
assumptions, Jansen and Zhang~\cite{jansen05} derive a
3.29~approximation algorithm for arbitrary precedence constraints,
which is improved in a 2.62~approximation in the particular case of a
series-parallel precedence graph by Lepere et
al.~\cite{lepere02}. However, although polynomial, these algorithms
relies on complex optimization techniques, which makes them difficult
to implement in a practical setting.

In this study, we consider a special case of malleable tasks, where
the speedup function of each task is $p^\alpha$, where $p$ is the
number of processors allocated to the task, and $0<\alpha\leq 1$ is a
global parameter.  In particular, when the share of processors $p_i$
allocated to a task $T_i$ is constant, its processing time is given by
$L_i/ p_i^\alpha$, where $L_i$ is the sequential duration of
$T_i$. The case $\alpha=1$ represents the unrealistic case of a
perfect linear speed-up, and we rather concentrate on the case
$\alpha<1$ which takes into consideration the cost of the
parallelization. In particular $\alpha<1$ accounts for the cost of
intra-task communications, without having to decompose the tasks in
smaller granularity sub-tasks with explicit communications, which
would make the scheduling problem intractable. This model has been
advocated by Prasanna and Musicus~\cite{prasmus2} for matrix
operations, and we present some new motivation for this model in our
context. As in~\cite{prasmus2}, we also assume that it is possible to
allocate non-integer shares of processors to tasks. This amounts to
assume that processors can share their processing time among
tasks. When task $A$ is allocated 2.6 processors and task $B$ 3.4
processors, one processor dedicates 60\% of its time to $A$ and 40\%
to $B$. Note that this is a realistic assumption, for example, when
using modern task-based runtime systems such as StarPU~\cite{starpu},
KAAPI~\cite{kaapi}, or PaRSEC~\cite{parsec}. This allows to simplify
the scheduling problem and to derive optimal allocation algorithms.

Our objective is to minimize the total processing time of a tree of
malleable tasks. Initially, we consider a homogeneous platform
composed of $p$ identical processors. To achieve our goal, we take
advantage of two sources of parallelism: the \emph{tree parallelism}
which allows tasks independent from each others (such as siblings) to
be processed concurrently, and the \emph{task parallelism} which
allows a task to be processed on several processors. A solution to
this problem describes both in which order tasks are processed and
which share of computing resources is allocated to each task.

In~\cite{prasmus,prasmus2}, the same problem has been addressed by
Prasanna and Musicus for series-parallel graphs (or SP-graphs). Such
graphs are built recursively as series or parallel composition of two
smaller SP-graphs. Trees can be seen as a special-case of
series-parallel graphs, and thus, the optimal algorithm proposed
in~\cite{prasmus,prasmus2} is also valid on trees. They use optimal
control theory to derive general theorems for any strictly increasing
speedup function. For the particular case of the speedup function
$p^\alpha$, Prasanna and Musicus prove some properties of the unique
optimal schedule which allow to compute it efficiently. Their results
are powerful (a simple optimal solution is proposed), but to obtain
these results they had to transform the problem in a shape which is
amenable to optimal control theory. Thus, their proofs do not provide
any intuition on the underlying scheduling problem, yet it seems
tractable using classic scheduling arguments.

In this paper, our contributions are the following:
\begin{compactitem}
\item In Sect.~\ref{sec:motiv}, we show that the model of malleable
  tasks using the $p^\alpha$ speed-up function is justified in the
  context of sparse matrix factorization.
\item In Sect.~\ref{sec:shared}, we propose a new and simpler proof
  for the results of ~\cite{prasmus,prasmus2} on series-parallel
  graphs, using pure scheduling arguments. 
\item In Sect.~\ref{sec:dist}, we extend the previous study on
  distributed memory machines, where tasks cannot be distributed
  across several distributed nodes. We provide NP-completeness
  results and approximation algorithms.
\end{compactitem}

%\section{Application and platform model}

\section{Validation of the Malleable Task Model}
\label{sec:motiv}

In this section, we evaluate the model proposed by
Prasanna and Musicus in~\cite{prasmus,prasmus2} for our target
application. This model states that the instantaneous speedup of a task processed on
$p$ processors is $p^\alpha$. Thus, the processing time of a task
$T_i$ of size $L_i$ which is allocated a share of processors $p_i(t)$
at time $t$ is equal to the smallest value $C_i$ such that
$
{ \int_{0}^{C_i} \left(p_i(t)\right)^{\alpha} dt} \ \geq\  L_i,
$
where $\alpha$ is a task-independent constant. When the share
of processors $p_i$ is constant, $C_i = L_i/p_i^\alpha$. Our goal is
(i) to find whether this formula well describes the evolution of the
task processing time for various shares of processors and (ii) to
check that different tasks of the same application have the same
$\alpha$ parameter.  We target a modern multicore platform composed of
a set of nodes each including several multicore processors. For the
purpose of this study we restrict ourselves to the single node case
for which the communication cost will be less dominant. In this
context, $p_i(t)$ denotes the number of \emph{cores} dedicated to task
$T_i$ at time $t$.

We consider applications having a tree-shaped task
graph constituted of parallel tasks. This kind of
execution model can be met in sparse direct solvers where the matrix
is first factorized before the actual solution is computed. For
instance, either the multifrontal method~\cite{dure:83} as implemented
in \mumps~\cite{mumps11} or \qrm~\cite{buttari10}, or the supernodal
approach as implemented in \superlu~\cite{superlu} or
in \pastix~\cite{pastix}, are based on tree-shaped task graphs (namely the
assembly tree~\cite{aglp:87}).  Each task in this tree is a partial
factorization of a dense sub-matrix or of a sparse panel. In order to
reach good performance, these factorizations are performed using tiled
linear algebra routines (BLAS): the sub-matrix is decomposed into 2D
tiles (or blocks), and optimized BLAS kernels are used to perform the
necessary operations on each tile. Thus, each task can be seen as a
task graph of smaller granularity sub-tasks.
%, which we call \emph{kernels} to avoid confusion. 
% See Figure~\ref{fig.ex-task-dags} for an illustration.

% \begin{figure}[htbp]
%   \centering
%   \subfigure[Tiled dense sub-matrix to be partially decomposed.]{%
%     \label{fig.ex-task-dags.matrix}%
%     \hspace{2cm}%
%     \includegraphics[scale=0.3]{partial-cholesky-5-2-matrix.fig}%
%     \hspace{2cm}%
%   }
  
%   \subfigure[Corresponding kernel graph.]{%
%     \label{fig.ex-task-dags.kernels}%
%     \includegraphics[scale=0.17]{partial-cholesky-5-2.fig}%
%   }
%   \caption{Example of the decomposition of a task of the DAG of a
%     Cholesky decomposition into smaller kernels.}
%   \label{fig.ex-task-dags}
% \end{figure}

As computing platforms evolve quickly and become more complex (e.g.,
because of the increasing use of accelerators such as GPUs or Xeon
Phis), it becomes interesting to rely on an optimized dynamic runtime
system to allocate and schedule tasks on computing resources. These
runtime systems (such as StarPU~\cite{starpu}, KAAPI~\cite{kaapi}, or
PaRSEC~\cite{parsec}) are able to process a task on a prescribed
subset of the computing cores that may evolve over time. This
motivates the use of the malleable task model, where the share of
processors allocated to a task vary with time. This approach has been
recently used and evaluated~\cite{hugo14} in the context of the \qrm
solver using the StarPU runtime system.

In order to assess whether tasks used within sparse direct
solvers fit the model introduced by Prasanna and Musicus
in~\cite{prasmus2} we conducted an experimental study on several dense
linear algebra tasks. We used a test platform composed of 4 Intel
E7-4870 processors having 10 cores each clocked at 2.40~GHz and having
30~MB of L3 cache for a total of 40 cores. The platform is equipped
with 1~TB of memory with uniform access. We considered dense
operations which are representative of what can be met in sparse
linear algebra computations, namely the standard frontal matrix
factorization kernel used in the \qrm solver. We used either
block-columns of size 32 (1D partitioning) or square blocks of size
256 (2D partitioning). All experiments were made using the StarPU
runtime.

% Figure~\ref{fig.qr4096} presents the timings obtained when computing
% the QR decomposition of a $M\times N$ matrix, with $M=4096$ and
% several values of $N$, for increasing number of processors. The
% logarithmic scales show that the $p^\alpha$ speedup function models
% well the timings, except for small matrices when $p$ is large. In this
% case, there is not enough parallelism in the task to exploit all
% available cores. We have performed a linear regression on the portion
% where $p\leq 10$ to compute the value of $\alpha$ for different task sizes. We
% performed similar experiments with a QR decomposition with $M=1024$,
% and for a Cholesky factorization. The obtained values of $\alpha$ are
% gathered in Figure~\ref{tab.alpha-dense}. All these
% values are very close to one, which means that the parallelization is
% almost perfect.

\begin{figure}[bt]
  \centering
  % \subfigure[Timings (points) and model (lines) for QR with $M=4096$]{\label{fig.qr4096}%
%     \includegraphics[width=0.45\linewidth]{sgeqrf_4096_model_log_log.fig}%
%   }
%   \hspace{1cm}
%   \raisebox{35pt}
%   {
%   \subfigure[\mystrut(1.2em)Values of $\alpha$ measured for QR and Cholesky]{\label{tab.alpha-dense}%
%     \scalebox{0.8}{\begin{tabular}[c]{|c|c|c|c|}
%       \hline
%       N & QR  & QR  &  Cholesky\\
%       & $M=1024$ &$M=4096$  &\\
%       \hline
%       5000 & 0.95 & 0.988 &  0.94 \\
%       \hline
%       10000 & 0.98 & 0.997 &  0.98 \\
%       \hline
%       15000 & 0.99 & 0.998 &  0.99 \\
%       \hline
%       20000 & 0.99 & 0.999 &  0.99 \\
%       \hline
%       25000 & 0.99 & 0.999 &  0.99 \\
%       \hline
%       30000 & 0.99 & 0.999 &  1.00 \\
%       \hline
%       35000 & 1.00 & 0.999 &  0.98 \\
%       \hline
%       40000 & 1.00 & 0.999 & 0.98 \\
%       \hline
%     \end{tabular}
%   }
%   }}\\
% 
\subfigure[\mystrut(1.2em) Timings and model (lines) with 1D partitioning]{\label{fig.qrm1D}%
  \includegraphics[width=0.6\linewidth]{qrm1d_model_log_log.fig}%
}
\hspace{0.5cm}
\raisebox{65pt}{%
\subfigure[\mystrut(1.2em)Values of $\alpha$]{\label{fig.qrmalpha}%
  \scalebox{0.9}{%
    \begin{tabular}[c]{|c|c|c|}
      \hline
      matrix & ~~1D~~ & ~~2D~~ \\
      \hline
      5000x1000~~  & 0.78 & 0.93\\
      \hline
      10000x2500~~ & 0.88 & 0.95 \\
      \hline
      20000x5000~~ & 0.89 & 0.94 \\
      \hline
    \end{tabular}}
}
}
\caption{Timings and $\alpha$ values for \qrm frontal matrix factorization kernel}
  \label{fig.qrm}
\end{figure}

Figure~\ref{fig.qrm1D} presents the timings obtained when processing
the \qrm frontal matrix factorization kernel on a varying number of
processors. The logarithmic scales show that the $p^\alpha$ speedup
function models well the timings, except for small matrices when $p$
is large. In those cases, there is not enough parallelism in tasks to
exploit all available cores. We performed linear regressions on
the portions where $p\leq 10$ to compute $\alpha$ for
different task sizes (Fig.~\ref{fig.qrmalpha}). We performed
the same test for 2D partitioning and computed the corresponding
$\alpha$ values (using $ p\leq 20$). We notice that the value of
$\alpha$ does not vary significantly with the matrix size, which
validates our model. The only notable exception is for the smallest
matrix (5000x1000) with 1D partitioning: it is hard to efficiently use
many cores for such small matrices. In all cases, when the number of
processors is larger than a threshold the performance deteriorates and
stalls. Our speedup model is only valid below this threshold, 
which threshold increases with the matrix size. This is not a
problem as the allocation schemes developed in the next sections
allocate large numbers of processors to large tasks at the top of the
tree and smaller numbers of processors for smaller tasks. In other
words, we produce allocations that always respect the validity
thresholds of the model.
 Finally, note that the value of $\alpha$ depends on the
 parameters of the problem (type of factorization, partitioning, block
 size, etc.). It has to be determined for each kernel and each set of
 blocking parameters. %before the execution %when
% considering a new kernel or new blocking parameters. The values of
% $\alpha$ obtained here are quite high thanks to the good memory
% performance of the considered computing platform. On other platforms,
% smaller values of $\alpha$ can be expected.

\section{Model and Notations}
\label{sec:model}

We assume that the number of available computing resources may vary
with time: $p(t)$ gives the (possibly rational) total number of processors
available at time $t$, also called the \fctofproc. For the sake of
simplicity, we consider that $p(t)$ is a step function. Although our
study is motivated by an application running on a single multicore
node (as outlined in the previous section), we use the term
\emph{processor} instead of \emph{computing core} in the following
sections for readability and consistency with the scheduling
literature.

We consider an in-tree $G$ of $n$ malleable tasks $T_1, \ldots, T_n$.
$L_i$ denotes the length, that is the sequential processing time, of
task $T_i$.  As motivated in the previous section, we assume that the
speedup function for a task allocated $p$ processors is $p^\alpha$,
where $0 < \alpha \leq 1$ is a fixed parameter.
% The processing time of task $T_i$ which is allocated $p_i(t)$
% processors at time $t$ is thus $\frac{L_i}{\int p_i(t) dt}$.
A schedule $\s$ is a set of nonnegative piecewise continuous functions
$\big\{ p_i(t)\ \big|\ i\in I\big\}$ representing the time-varying
share of processors allocated to each task. During a time interval
$\T$, the task $T_i$ performs an amount of work equal to $ \inte{\T\
}{p_i(t)^\alpha}{dt}$. Then, $T_i$ is completed when the total work
performed is equal to its length $L_i$. The completion time of task
$T_i$ is thus the smallest value $C_i$ such that $
\int_0^{C_i}{p_i(t)^\alpha}{dt} \geq L_i$.
We define $w_i(t)$ as the ratio of the
work of the task $T_i$ that is done during the time interval $[0,t]$:
$w_i(t) = \inte[t]{0}{p_i(x)^\alpha}{dx} \big/ L_i $.  A schedule is a
valid solution if and only if:
\begin{compactitem}
\item it does not use more processors than available: $\forall t,  \sum_{i\in I} p_i(t) \leq p(t)$;
\item it completes all the tasks: $\exists\tl,\ \forall i \in I, \ \ w_i(\tl)=1$;
\item and it respects precedence constraints:
  $\forall i\in I, \forall t$, if $p_i(t)>0$ then, $\forall j \in I$,
  if $j$ is a child of $i$, $w_j(t)=1$.
\end{compactitem}
The makespan $\tl$ of a schedule is computed as $\min \{t \ | \
\forall i\ w_i(t) = 1\}$. Our objective is to construct a valid
schedule with optimal, i.e., minimal, makespan.

% The ratio of processors given to a task is $p_i(t)/p(t)$. This quantity is relevant as it will be proved constant under certain conditions.

Note that because of the speedup function $p^\alpha$, the computations
in the following sections will make a heavy use of the functions
$f:x\mapsto x^\alpha$ and $g:x\mapsto x^{(1/\alpha)}$. We assume that
we have at our disposal a polynomial time algorithm to compute both
$f$ and $g$. We are aware that this assumption is very likely to be
wrong, as soon as $\alpha<1$, since $f$ and $g$ produce irrational
numbers. However, without these functions, it is not even possible to
compute the makespan of a schedule in polynomial time and, hence, the
problem is not in NP. Furthermore, this allows us to avoid the
complexity due to number computations, and to concentrate on the most
interesting combinatorial complexity, when proving NP-completeness
results and providing approximation algorithms. In practice, any
implementation of $f$ and $g$ with a reasonably good accuracy will be
sufficient to perform all computations including the computation
of makespans.

In the next section, following Prasanna and Musicus, we will not
consider trees but more general graphs: \emph{series-parallel graphs}
(or SP graphs). An SP graph is recursively defined as a single task,
the series composition of two SP graphs, or the parallel composition
of two SP graphs. % The two subgraphs in a parallel composition are
% called branches. Series compositions are ordered so that it is clear
% which subgraph should be executed first. The resulting DAG expresses
% the precedence constraints.
% Two nodes having the same set of
% predecessors, i.e., the same in-neighbors are called \emph{siblings}.
% %
A tree can easily be transformed into an SP graph by joining the
leaves according to its structure, % (see Figure \ref{fig:treeSP}),
the resulting graph is then called a \emph{pseudo-tree}.  We will use
$(\para ij)$ to represent the parallel composition of tasks $T_i$ and
$T_j$ and $(\seri ij)$ to represent their series composition. 
% The SP graph of Figure \ref{fig:treeSP} can be represented as
% $\left( \seri{
%   \left(\mystrut(1.1em)\para{\left(\mystrut(1em)\seri{\left(\para{\left(\para
%           45 \right)}{6}\right)\,}{\,2}\right)\,}{\,3}\right)\ }{\
%   1}\right)$.
Thanks to the construction of pseudo-trees, an algorithm which solves
the previous scheduling problem on SP-graphs also gives an optimal
solution for trees.

% \begin{figure}[htbp]
% 	\centering
%         \scalebox{0.7}{\input{fig/figspgraph.tex}}
% \caption{Example of a tree transformed into an SP graph.}
% \label{fig:treeSP}
% \end{figure}

\section{Optimal Solution for Shared-Memory Platforms}

\label{sec:shared}

% In this section, we do not focus on memory consumption but only on
% makespan minimization. Therefore, {\it optimal} should be understood
% here as {\it minimizing the makespan}.

% As explained above, the speedup is fixed in this section, which means that we design algorithms to a specific speedup, and we do not study the scheduling problem for generic speedup functions, as it was done in \cite{prasmus}.
% We will first study the model where $f(p) = p^\alpha$, as motivated in \cite{prasmus}, before designing a more realistic model for the special case where $p<1$.

The purpose of this section is to give a simpler proof of the results
of~\cite{prasmus,prasmus2} using only scheduling arguments. We
consider an SP-graph to be scheduled on a shared-memory platform
(each task can be distributed across the whole platform). We assume
that $\alpha<1$ and prove the uniqueness of 
 the optimal schedule.

Our objective is to prove that any SP graph $G$ is \emph{equivalent}
to a single task $T_G$ of easily computable length: for any \fctofproc
$p(t)$, graphs $G$ and $T_G$ have the same makespan. We prove that the
ratio of processors allocated to any task $T_i$, defined by
$\rat_i(t)= p_i(t)/p(t)$, is constant from the
moment at which $T_i$ is initiated to the moment at which it is
terminated. We also prove that in an optimal schedule, the two
subgraphs of a parallel composition terminate at the same
time and each receives a constant total ratio of processors throughout
its execution.
We then prove that these properties imply that the optimal schedule is unique and obeys
to a {\it flow conservation} property: the shares of processors
allocated to two subgraphs of a series composition are equal. When
considering a tree, this means that the whole schedule is defined by
the ratios of processors allocated to the leaves. Then, all the children of a node $T_i$ terminate at the
same time, and its ratio is the sum of its children
ratios.

We first need to define the length $\LG{G}$ associated to a graph $G$, which
will be proved to be the length of the task $T_G$. Then, we state a few lemmas
before proving the main theorem. We only present here sketches of the proofs,
the detailed versions can be found in \cite{RR-ipdps-2014}.

\begin{defi}
\label{def.eq-task}
We recursively define the length $\LG{G}$ associated to a SP graph $G$:
\begin{inparaitem}
	\item $\LG{T_i} = L_i$ \hfill 
	\item $\LG{\seri{G_1}{G_2}} = \LG{G_1} + \LG{G_2}$ \hfill
	\item $\LG{\para{G_1}{G_2}} = \left(\LG{G_1}^{1/\alpha} + \LG{G_2}^{1/\alpha}\right)^\alpha$
\end{inparaitem}
\end{defi}

\newcommand{\Q}{\mathcal{Q}}
%%%%%%%%%%%%%%%%%%%%%%%%%%%%%%%%%%%%%%%%%%%%%%%%%%%%%%%%%%%
%%%%%%%%%%%%%%%%%%%%%%%%%%%%%%%%%%%%%%%%%%%%%%%%%%%%%%%%%%%
%\subsubsection{General lemmas}

\begin{lemma}
  \label{lem:allproc}
  An allocation minimizing the makespan uses all the processors at any time. 
\end{lemma}

We call a \emph{clean interval} with regard to a schedule $\s$ an
interval during which no task is completed in $\s$.

\begin{lemma}
  \label{lem:esc}
  When the number of available processors is constant, any optimal schedule
  allocates a constant number of processors per task on any clean
  interval.
  %We have to process $n$ tasks in parallel with a constant number of processors $p$. A schedule that does not allocate a constant number of processors per task on clean intervals is not optimal.
\end{lemma}

% First complete proof, can be shrunk if necessary

\begin{proof}
  By contradiction, we assume that there exists an optimal schedule
  $\mathcal{P}$ of makespan $M$, a task $T_j$ and a clean interval
  $\T=[t_1,t_2]$ such that $T_j$ is not allocated a constant number of
  processors on $\T$. By definition of clean intervals, no task
  completes during $\T$. $|\T|=t_2-t_1$ denotes the duration of $\T$,
  $I$ the set of tasks that receive a non-empty share of processors
  during $\T$, and $p$ the constant number of available processors.
	
%   By contradiction, we consider an optimal schedule $\mathcal{P}$ of
%   makespan $M$, and we suppose that one task $T_j$ is not allocated a
%   constant number of processors 
% %  \footnote{Formally, an allocation is
% %    constant on $\T$ if it is equal to a given constant except, maybe,
% %    on a subset of $\T$ of null measure.}
%     on a clean interval
%   $\T=[t_1,t_2]$. By definition of clean intervals, no task completes
%   during $\T$. $|\T|=t_2-t_1$ denotes the duration of $\T$, $I$ 
%   the set of tasks that receive a non-empty share of processors during
%   $\T$, and $p$ the constant number of available processors.
	
  % From now on, the index $j$ will refer to this particular task and
  % the index $i$ will refer to any task. Let $p_i(t)$ denote the share
  % of processors allocated to task $T_i$ at time $t\in \T$.
	
  We want to show that there exists a valid schedule with a makespan
  smaller than $M$. To achieve this, we define an intermediate and
  not necessarily valid schedule $\Q$, which
  nevertheless respects the resource constraints (no more than
  $p$ processors are used at time $t$). This schedule is
  equal to $\mathcal{P}$ except on $\T$.
  The constant share of processors allocated to task $T_i$ on $\T$ in
  $\mathcal{Q}$ is defined by $q_i = \frac{1}{|\T|}\int_\T
  p_i(t)dt$.  For all $t$, we have $\sum_{i\in I} p_i(t) = p$ because
  of Lemma~\ref{lem:allproc}. We get $\sum_{i\in I} q_i = p$. So
  $\mathcal{Q}$ respects the resource constraints.
  Let $W_i^\T(\mathcal{P})$ (resp. $W_i^\T(\mathcal{Q})$) denote the
  work done on $T_i$ during $\T$ under schedule $\mathcal{P}$
  (resp. $\mathcal{Q}$).
    We have 
  \begin{align*}	
    W_i^\T(\mathcal{P}) &= \int_\T p_i(t)^\alpha dt = |\T|
    \int_{[0,1]}  p_i(t_1+ t |\T|)^\alpha dt\\
  % \end{align*}
  % \begin{align*}
    W_i^\T(\mathcal{Q}) &= \int_\T  \left(\frac{1}{|\T|}\int_\T p_i(t) dt\right) ^\alpha dx
    =  |\T| \left(\int_{[0,1]} p_i(t_1+t |\T|) dt\right) ^\alpha
  \end{align*}
  As $\alpha<1$, the function $x\mapsto x^\alpha$ is concave and then,
  by Jensen inequality \cite{Hardy}, $W_i^\T(\mathcal{P}) \leq
  W_i^\T(\mathcal{Q})$.  Moreover, as $x\mapsto x^\alpha$ is
  \emph{strictly} concave, this inequality is an equality if and only
  if the function $t\mapsto p_i(t_1+t|\T|)$ is equal to a constant on
  $[0,1[$ except on a subset of $[0,1[$ of null measure \cite{Hardy}.
  Then, by definition, $p_j$ is not constant on $\Delta$, and
  cannot be made constant by modifications on a set of null measure.
  We thus have $W_j^\T(\mathcal{P}) < W_j^\T(\mathcal{Q})$.
  Therefore, $T_j$ is allocated too many processors under
  $\mathcal{Q}$.  It is then possible to distribute this surplus among
  the other tasks during $\T$, so that the work done during $\T$ in
  $\mathcal{P}$ can be terminated earlier. This remark implies that
  there exists a valid schedule with a makespan smaller than $M$;
  hence, the contradiction.\qed
\end{proof}

We recall that  $\rat_i(t)= p_i(t)/p(t)$ is the instantaneous ratio of
processors allocated to a task $T_i$ .

\begin{lemma}
\label{lem:rate}
Let $G$ be the parallel composition of two tasks, $T_1$ and $T_2$. If
$p(t)$ is a step function, in any optimal schedule $r_1(t)$ is
constant and equal to $\pi_1 =
{1}\left/\left({1+\left({L_2}/{L_1}\right)^
      {1/\alpha}}\right)\right. = L_1^{1/\alpha} \left/ \LG{\para
    12}^{1/\alpha} \right.$ up to the completion of $G$.
\end{lemma}

\begin{proof}
First, we prove that $r_1(t)$ is constant on any optimal
schedule.

We consider an optimal schedule $\s$, and two consecutive time intervals $A$ and
$B$ such that $p(t)$ is constant and equal to $p$ on $A$ and $q$ on $B$, and
$\s$ does not complete before the end of $B$. Suppose also that $|A|p^\alpha =
|B|q^\alpha$ (shorten one interval otherwise), where $|A|$ and
$|B|$ are the durations of intervals $A$ and $B$. By Lemma \ref{lem:esc}, 
$r_1(t)$ has constant values $r_1^A$ on $A$ and $r_1^B$ on
$B$. Suppose by contradiction that $r_1^A \neq r_1^B$.% and without
                                % loss of generality that $r_1^A <
                                % r_1^B$. % not needed here

We want to prove that $\s$ is not optimal, and so that we can do the
same work as $\s$ does on $A\cup B$ in a smaller makespan. We
set $r_1 = \left.\left({r_1^A+r_1^B}\right) \right/{2}$.  We define
the schedule $\s'$ as equal to $\s$ except on $A \cup B$ where the ratio
allocated to $T_1$ is $r_1$ (see Fig. \ref{fig:rarb}).
\begin{figure}[tb]
\centering
\ifdefined\s
\else
\newcommand{\s}{S}
\fi

\newcommand{\ya}{30pt}
\newcommand{\yb}{50pt}

\newcommand{\yf}{.5*\ya+.5*\yb}

\newcommand{\topy}{70pt}
\newcommand{\topx}{100pt}
\newcommand{\midx}{60pt}

\newcommand{\midlx}{50pt}

\newcommand{\scalerat}{0.8}
\begin{tikzpicture}[scale =\scalerat,
every node/.style={inner sep=0pt, minimum size=0pt, anchor=center, transform shape},
label/.style={}
]

\node (bl) at (0,0) {};
\node (tl) at (0,\topy) {};
\node (br) at (\topx,0) {};
\node (tr) at (\topx,\topy) {};
\node (bm) at (\midx,0) {};
\node (tm) at (\midx,\topy) {};

\node (mra) at (\midx,\ya) {};
\node (mla) at (0,\ya) {};

\node (mlb) at (\midx,\yb) {};
\node (mrb) at (\topx,\yb) {};

\node (mlf) at (0,\yf) {};
\node (mrf) at (\topx,\yf) {};

\node (title) at (\topx/2, \topy+15pt) {$\s$};

\draw (bl)--(tl)--(tr)--(br)--(bl);
\draw (mla)--(mra)
      (mlb)--(mrb)
      (bm)--(tm);

\path (tl)--+(10pt,-10pt) node[label] {$T_2$};
\path (0,\ya)--+(10pt,-10pt) node[label] {$T_1$};
\path (\midx,\yb)--+(10pt,-10pt) node[label] {$T_1$};
\path (tm)--+(10pt,-10pt) node[label] {$T_2$};
\path (mla)--+(-10pt,-0) node[label] {$r_1^A$};
\path (mrb)--+(+10pt,-0) node[label] {$r_1^B$};
\draw[dotted] (tl)--+(-15pt,-0) node[left] {$1$};
\draw[dotted] (bl)--+(-15pt,-0) node[left] {$0$};

\draw [decorate,decoration={brace,amplitude=5pt},yshift=-2pt,xshift=0pt]
(\midx,0)--(0,0) node [midway,yshift=-0.4cm] 
{ $A$};
\draw [decorate,decoration={brace,amplitude=5pt},yshift=-2pt,xshift=0pt]
(\topx,0)--(\midx,0) node [midway,yshift=-0.4cm] 
{$B$};

\end{tikzpicture}
\raisebox{40pt}{%\hspace{10pt}
\huge$\Rightarrow$%\hspace{5pt}
}
%\raisebox{-3pt}
{\begin{tikzpicture}[scale =\scalerat,
every node/.style={inner sep=0pt, minimum size=0pt, anchor=center, transform shape},
label/.style={}
]

\draw (bl)--(tl)--(tr)--(br)--(bl);
\draw (mlf)--(mrf)
      (bm)--(tm);

\path (tl)--+(10pt,-10pt) node[label] {$T_2$};
\path (0,\yf)--+(10pt,-10pt) node[label] {$T_1$};
\path (\midx,\yf)--+(10pt,-10pt) node[label] {$T_1$};
\path (tm)--+(10pt,-10pt) node[label] {$T_2$};
\path (mlf)--+(-10pt,-0) node[label] {$r_1$};
\draw[dotted] (tl)--+(-15pt,-0) node[left,label] {$1$};
\draw[dotted] (bl)--+(-15pt,-0) node[left,label] {$0$};

\path (title) node[label] {$\s'$};

\draw [decorate,decoration={brace,amplitude=5pt},yshift=-2pt,xshift=0pt]
(\midx,0)--(0,0) node [midway,yshift=-0.4cm] 
{ $A$};
\draw [decorate,decoration={brace,amplitude=5pt},yshift=-2pt,xshift=0pt]
(\topx,0)--(\midx,0) node [midway,yshift=-0.4cm] 
{$B$};

\end{tikzpicture}}
\caption{Schedules $\s$ and $\s'$ on $A\cup B$. The abscissae represent the time and the ordinates the ratio of processing power}
\label{fig:rarb}
\end{figure}
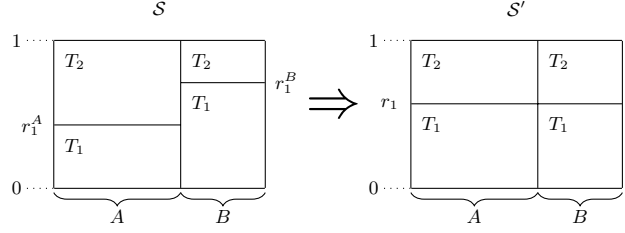
\\
The work $W_1$ on task $T_1$ under $\s$ and $W'_1$ under $\s'$ during $A\cup B$ are
equal to:
$$W_1 = |A|p^\alpha \left(r_1^A\right)^\alpha + |B|q^\alpha \left(r_1^B\right)^\alpha
\qquad W'_1 = r_1^\alpha \left(|A|p^\alpha + |B|q^\alpha\right)$$
Then, with the concavity inequality and the fact that $|B|q^\alpha =
|A|p^\alpha$, we can deduce that $W_1'>W_1$ and symmetrically that $W_2'>W_2$.

Therefore, $\s'$ performs strictly more work for each task during $A\cup B$ than
$\s$. Thus, as in Lemma~\ref{lem:esc}, $\s$ is not optimal. So
$r_1(t)$ is constant in optimal schedules.

There remains to prove that in an optimal schedule $\s$, $r_1(t) =
\pi_1$; hence, the optimal schedule is unique.  As $p(t)$ is a step
function, we define the sequences $\left(A_k\right)$ and
$\left(p_k\right)$ such that $A_k$ is the duration of the $k$-th
step of the function $p(t)$ and $p(t)=p_k>0$ on $A_k$. The sum of the
durations of the $A_k$'s is the makespan of $\s$.
Then, if we note $V = \sum_k |A_k| p_k^\alpha$ and $r_1$ the value of $r_1(t)$,
we have:
$$
\begin{array}{c}\displaystyle
L_1 = \sum_k |A_k| r_1^\alpha p_k^\alpha = r_1^\alpha V
\qquad\displaystyle\text{and}\qquad 
 L_2 = \sum_k |A_k| (1-r_1)^\alpha p_k^\alpha = (1-r_1)^\alpha V
\end{array}
$$
% \begin{align*}
% L_1 &= \sum_i |A_i| r_1^\alpha p_i^\alpha = r_1^\alpha V 
% \end{align*}
% \begin{align*}
% \text{and}\qquad 
%  L_2 = \sum_i A_i (1-r_1)^\alpha p_i^\alpha = (1-r_1)^\alpha V
% \end{align*}
% Then,
% \begin{align*}
% L_2 &= (1-r_1)^\alpha \frac{L_1}{r_1^\alpha}\qquad
% \text{thus}\qquad r_1 = \frac{1}{1+\left(\frac{L_2}{L_1}\right)^ {1/\alpha}} = \pi_1
% \end{align*}
Then, $\ r_1 = {1}\left/\left({1+\left({L_2}/{L_1}\right)^
      {1/\alpha}}\right)\right.  = \pi_1$.\qed
\end{proof}

\begin{lemma}
  \label{lem:ratequiv}
  Let $G$ be the parallel composition of tasks $T_1$ and $T_2$, with
  $p(t)$ a step function, and $\s$ an optimal schedule. Then, the
  makespan of $G$ under $\s$ is equal to the makespan of the task
  $T_G$ of length $\LG G = \LG{\para 12}$.
\end{lemma}

\begin{proof}  
  We characterize $p(t)$ by the sequences $(A_k)$ and $(p_k)$ as in the proof of
  Lemma~\ref{lem:rate}. We know by Lemma~\ref{lem:rate} that the share allocated
  to $T_1$ is constant and equal to $\pi_1p_k$ on each interval $A_k$.
  Then, by summing the work done on each interval for both tasks, one can prove
  that they are completed simultaneously, and that this completion time is the
  same as that of task $T_G$ under the same processor profile.\qed
\end{proof}

\begin{theorem}
  \label{th:step}
  For every graph $G$, if $p(t)$ is a step function, $G$ has the same
  optimal makespan as its equivalent task $T_G$ of length $\LG
  G$ (computed as in Definition~\ref{def.eq-task}). Moreover, there is a unique optimal schedule, and it can be
  computed in polynomial time.
\end{theorem}

\begin{proof}
  In this proof, we only consider optimal schedules. Therefore, when
  the makespan of a graph is considered, this is implicitly its
  optimal makespan.
  We first remark that in any optimal schedule, as $p(t)$ is a step
  function and because of Lemma \ref{lem:esc}, only step functions are
  used to allocate processors to tasks, and so Lemma
  \ref{lem:ratequiv} can be applied on any subgraph of $G$ without
  checking that the \fctofproc is also a step function for this
  subgraph.
  We now prove the result by induction on the structure of $G$.
  \begin{itemize}
  \item $G$ is a single task. The result is immediate.

  \item $G$ is the series composition of $G_1$ and $G_2$. By
    induction, $G_1$ (resp. $G_2$) has the same makespan as task
    $T_{G_1}$ (resp. $T_{G_2}$) of length $\LG {G_1}$ (resp. $\LG {G_2}$) under any
    \fctofproc. Therefore, the makespan of $G$ is equal to $\LG G = \LG {\seri {G_1}{G_2}} =
    \LG {G_1} + \LG {G_2}$.
    The unique optimal schedule of $G$ under $p(t)$ processors is the
    concatenation of the optimal schedules of $G_1$ and $G_2$.

  \item $G$ is the parallel composition of $G_1$ and $G_2$. By
    induction, $G_1$ (resp. $G_2$) has the same makespan as task
    $T_{G_1}$ (resp. $T_{G_2}$) of length $\LG {G_1}$ (resp. $\LG {G_2}$) under any
    \fctofproc.
    Consider an optimal schedule $\s$ of $G$ and let $p_1(t)$ be the
    \fctofproc allocated to $G_1$. Let $\tilde \s$ be the schedule of
    $(\para{T_{G_1}}{T_{G_2}})$ that allocates $p_1(t)$ processors to
    $T_{G_1}$. $\tilde \s$ is optimal and achieves the same makespan as $\s$
    for $G$ because $T_{G_1}$ and $G_1$ (resp. $T_{G_2}$ and $G_2$) have the
    same makespan under any \fctofproc.  Then, by Lemma
    \ref{lem:ratequiv}, $\tilde \s$ (so $\s$) achieves the same makespan as
    the optimal makespan of the task $T_G$ of length $\LG {\para{G_1}{G_2}} = \LG G$.
    Moreover, by Lemma \ref{lem:rate} applied on $(\para{T_{G_1}}{T_{G_2}})$,
    we have $p_1(t) = \pi_1 p(t)$. By induction, the unique optimal
    schedules of $G_1$ and $G_2$ under respectively $p_1(t)$ and
    $(p(t)-p_1(t))$ processors can be computed. Therefore, there is a
    unique optimal schedule of $G$ under $p(t)$ processor: the
    parallel composition of these two schedules.  
  \end{itemize}
  
  Therefore, there is a unique optimal schedule for $G$ under $p(t)$. Moreover,
  it can be computed in polynomial time. We describe here the algorithm to
  compute the optimal schedule of a tree $G$, but
  it can be extended to treat SP-graphs. The length of the equivalent
  task of each subtree of $G$ can be computed in polynomial time by a
  depth-first search of the tree (assuming that raising a number to the power
  $\alpha$ or $1/\alpha$ can be done in polynomial time). Hence, the ratios
  $\pi_1$ and $\pi_2$ for each parallel composition can also be computed in
  polynomial time. Finally, these ratios imply the computation in linear time of
  the ratios of the processor profile that should be allocated to each task
  after its children are completed, which describes the optimal schedule.\qed
\end{proof}

%/////////////////////////////////////////
%// TODO: REMOVE?
%/////////////////////////////////////////
%
%{\color{red}
%In order to obtain Theorem~\ref{th:step}, we needed to use the value of the speedup function: $f:p\mapsto p^\alpha$ for a fixed $\alpha$. Actually, the shape of the optimal solution as described in the theorem is only valid for such functions. Indeed, let the \emph{scale property} denotes the first statement of Lemma~\ref{lem:rate} (i.e., $r_1(t)$ is independent of $p(t)$ in optimal schedules on $\para{T_1}{T_2}$). Then, Lemma~\ref{lem:scaleprop} proved in \ref{NEW REPORT} states:
%
%\begin{lemma}
%\label{lem:scaleprop}
% The only valid strictly convave and increasing speedup functions that lead to optimal schedules respecting the scale property are of the shape $p\mapsto p^\alpha$.
% \end{lemma}
%}

\section{Extensions to Distributed Memory}
\label{sec:dist}

The objective of this section is to extend the previous results to the
case where the computing platform is composed of several nodes with
their own private memory. In order to avoid the large communication overhead
of processing a task on cores distributed across several
nodes, we forbid such a multi-node execution: the tasks of the tree can
be distributed on the whole platform but each task has to be processed on a
single node. We prove that this additional constraint, denoted by
\R, renders the problem much more difficult. We concentrate first
on platforms with two homogeneous nodes and then with two heterogeneous
nodes.

\subsection{Two Homogeneous Multicore Nodes}
\label{sec:dist-hom}

In this section, we consider a multicore platform composed of two
equivalent nodes having the same number of computing cores $p$. We
also assume that all the tasks $T_i$ have the same speedup function
$p_i^\alpha$ on both nodes.
We first show that finding a schedule with minimum makespan is weakly
NP-complete, even for independent tasks:

\begin{theorem}
  Given two homogenous nodes of $p$ processors, $n$ independent tasks
  of sizes $L_1, ..., L_n$ and a bound $T$, the problem of finding a
  schedule of the $n$ tasks on the two nodes that respects \R, and
  whose makespan is not greater than $T$, is (weakly) NP-complete for
  all values of the $\alpha$ parameter defining the speedup function.
\end{theorem}

The proof relies on the Partition problem, which is known to be weakly
(i.e., binary) NP-complete~\cite{gareyjohnson}, and uses tasks of
length $L_i=a_i^\alpha$, where the $a_i$'s are the numbers from the
instance of the Partition problem. We recall that we assume that
functions $x\mapsto x^\alpha$ and $x\mapsto x^{1/\alpha}$ can be
computed in polynomial time.  
\iflong
\todo[inline]{faire une courte preuve de
  NP-complétude}
\else
Details can be found in the companion research report~\cite{RR-ipdps-2014}.
\fi

We also provide a constant ratio approximation algorithm. We recall
that a $\rho$-approximation provides on each instance a solution whose
objective $z$ is such that $z \leq \rho z*$, where $z*$ is the optimal
value of the objective on this instance.

\begin{theorem}
\label{th.43approx}
  There exists a polynomial time \frtrd-approximation algorithm for
  the makespan minimization problem when scheduling a tree of malleable tasks on
  two homogenous nodes.
\end{theorem}

%%%%%%%%%% next paragrpah is only for the IPDPS version
\iflong\else Due to lack of space, we refer the interested reader to
the companion research report
for the complete description of the
algorithm and proof~\cite{RR-ipdps-2014}.  The proof of the
approximation ratio consists in comparing the proposed solution to the
optimal solution on a single node made of $2p$ processors, denoted
\spm. Such an optimal solution can be computed as proposed in the
previous section, and is a lower bound on the optimal makespan on 2
nodes with $p$ processors. The general picture of the proposed
algorithm is the following. First, the root of the tree is arbitrarily
allocated to
the $p$ processors of one of the two nodes. Then, the subtrees $S_i$'s
rooted at the root's children are considered. If none of these
subtrees is allocated more than $p$ processors in \spm, then we show
how to ``pack'' the subtrees on the two nodes and bound the slow-down
by \frtrd. On the contrary, if one of the $S_i$'s is allocated more
than $p$ processors in \spm, then we allocate $p$ processors to its
root, and recursively call the algorithm on  its children and on the
remaining subtrees.
\fi

%%%%%%%%%% what follows in only for the RR
% iflong

\subsection{Two Heterogeneous Multicore Nodes}
\label{sec:dist-het}

We suppose here that the computing platform is made of two processors
of different processing capabilities: the first one is made of $p$
cores, while the second one includes $q$ cores. We also assume that
the parameter $\alpha$ of the speedup function is the same on both
processors. As the problem gets more complicated, we concentrate here
on $n$ independent tasks,  of lengths $L_1, ..., L_n$. Thanks to the
homogenous case presented above, we already know that scheduling
independent tasks on two nodes is NP-complete.

\iflong \else %%%%%%%%%%%%%%%%% Version IPDPS
This problem is close to the \textsc{Subset Sum} problem. Given $n$
numbers, the optimization version of \textsc{Subset Sum} considers a
target $K$ and aims at finding the subset with maximal sum smaller
than or equal to $K$.  There exists many approximation schemes for
this problem. In particular, Kellerer et
al.~\cite{kellerer2003efficient} propose a fully polynomial
approximation scheme (FPTAS). Based on this result, an approximation
scheme can be derived for our problem.
% This problem is close to the \textsc{Subset Sum} problem which given
% $n$ numbers, consists in finding a subset of them which sums to
% zero. The optimization version of \textsc{Subset Sum} also considers a
% target $K$ and aims at finding the subset with maximal sum smaller or
% equal to $K$.  There exists many approximation schemes for this
% problem. In particular, Kellerer et al.~\cite{kellerer2003efficient}
% propose a fully polynomial approximation scheme (FPTAS). Based on this
% result, an approximation scheme can be derived for our problem.
\begin{theorem}
  There exists an FPTAS for the problem of scheduling independent
  malleable tasks on two heterogeneous nodes, provided that, for each
  task, $L_i^{1/\alpha}$ is an integer.
\end{theorem}

The proof is complex and detailed in~\cite{RR-ipdps-2014}. The
assumption on the $L_i^{1/\alpha}$s is needed to apply the FPTAS of
\textsc{Subset Sum}, which is valid only on integers.
\fi %%%%%%%%%%%%%%%%%%%% fin version IPDPS

%%%%%%%% fin version RR

\section{Conclusion}

In this paper, we have studied how to schedule trees of malleable
tasks whose speedup function on multicore platforms is $p^\alpha$. We
have first motivated the use of this model for sparse matrix
factorizations by actual experiments. When using factorization kernels
actually used in sparse solvers, we show that the speedup follows the
$p^\alpha$ model for reasonable allocations. On the machine used for
our tests, $\alpha$ is in the range 0.85--0.95. Then, we proposed a
new proof of the optimal allocation derived by Prasanna and
Musicus~\cite{prasmus,prasmus2} for such trees on single node
multicore platforms. Contrarily to the use of optimal control theory
of the original proofs, our method relies only on pure scheduling
arguments and gives more intuitions on the scheduling problem. Based
on these proofs, we proposed several extensions for two multicore
nodes: we prove the NP-completeness of the scheduling problem and propose
a \frtrd-approximation algorithm for a tree of malleable tasks on two
homogeneous nodes, and an FPTAS for independent malleable tasks on two
heterogeneous nodes. 
%Finally, we have estimated the potential gain of
%using an optimal allocation compared to simpler allocations from the
%literature on a single multicore node by extensive simulations.
%Although the improvement over simpler allocations may seem small in
%the measured range of $\alpha$ values, it has to be noted that (i)
%even a 5\% is interesting when comparing real software
%implementations, which is also why the PM allocation has already been
%considered for sparse solvers~\cite{guermouche}, (ii) the value of
%$\alpha$ is expected to be smaller for machine with weaker memory
%bandwidth and (iii) memory bandwidth increases at a smaller pace than
%core computing rates~\cite{graham2005getting}, which makes smaller
%values of $\alpha$ more relevant.

The perspectives to extend this work follow two main
directions. First, it would be interesting to extend the
approximations proposed for the heterogeneous case to a number of nodes
larger than two, and to more heterogeneous nodes, for which the value
of $\alpha$ differs from one node to another. This is a promising
model for the use of accelerators (such as GPU or Xeon Phi). The
second direction concerns an actual implementation of the PM
allocation scheme in a sparse solver. % Our last results show that the
% large implementation effort necessary for such an experimentation is
% worth it, especially when considering multicore platforms with many
% cores and limited memory bandwidth.

\bibliographystyle{splncs03}
\bibliography{europar}
\end{document}